\documentclass[a4paper,UKenglish,cleveref, autoref, thm-restate]{lipics_arxiv}
\hideLIPIcs  
\nolinenumbers 

\author{Alexandra Weinberger}{FernUniversit\"at in Hagen, Germany}{alexandra.weinberger@fernuni-hagen.de}{https://orcid.org/0000-0001-8553-6661}{}

\author{Ji Zeng}{Alfréd Rényi Institute of Mathematics, Hungary}{zeng.ji@renyi.hu}{https://orcid.org/0009-0002-7240-4832}{Supported by ERC Advanced Grants ``GeoScape'', no. 882971 and ``ERMiD'', no. 101054936, and also partly supported by NSF grant DMS-1928930 while in residence at the Simons--Laufer Mathematical Sciences Institute during the Spring 2025 semester.}

\authorrunning{A. Weinberger and J. Zeng} 

\Copyright{Alexandra Weinberger and Ji Zeng} 

\title{What induces plane structures in complete graph drawings?}

\acknowledgements{We wish to thank Jeck Lim, Jonathan Rollin, and André Schulz for helpful discussions, and an anonymous reviewer who provided the idea to simplify Lemma~\ref{lem:samecell}.}

\begin{document}

\maketitle

\begin{abstract}
This paper considers the task of connecting points on a piece of paper by drawing a curve between each pair of them. Under mild assumptions, we prove that many pairwise disjoint curves are unavoidable if either of the following rules is obeyed: any two adjacent curves do not cross, or any two non-adjacent curves cross at most once. Here, two curves are called adjacent if they share an endpoint.
On the other hand, we demonstrate how to draw all curves such that any two adjacent curves cross exactly once, any two non-adjacent curves cross at least once and at most twice, and thus no two curves are disjoint.
Furthermore, we analyze the emergence of disjoint curves without these mild assumptions, and characterize the plane structures in complete graph drawings guaranteed by each of the rules above.
\end{abstract}

\section{Introduction}\label{sec:introduction}
In the plane, a \textit{complete graph drawing} is a collection of points and curves connecting each pair of these points. The points are called vertices and the curves are called edges of this drawing. More precisely, an edge is the image of a piecewise smooth map from a closed interval on the real line to the plane which starts and ends at its corresponding pair of vertices. In this paper, we make \textit{mild assumptions} for complete graph drawings: we assume no edge contains vertices other than its endpoints; we assume no edge self-intersects, that is, its defining map must be injective; we assume edges cross properly whenever they intersect, in particular, edges do not merely touch. It is easy to imagine that a complete graph drawing can be tangled like a mass of spaghetti noodles. Our main result asserts that the breaking point of ``tangledness'' in complete graph drawings is the logical disjunction of the two rules below. Here, two edges are called adjacent if they share an endpoint.
\begin{equation*}
    \begin{aligned}
    &\text{Adjacent-simple: any two adjacent edges do not cross;}\\
    &\text{Separate-simple: any two non-adjacent edges cross at most once.}
\end{aligned}
\end{equation*}

\begin{theorem}\label{thm:main1}
    For any positive integer $m$, there exists $n$ such that every complete graph drawing on $n$ vertices that is either adjacent-simple or separate-simple must contain $m$ pairwise disjoint edges. On the other hand, there exist complete graph drawings on arbitrarily many vertices in which any two adjacent edges cross exactly once, and any two non-adjacent edges cross at least once and at most twice.
\end{theorem}

\begin{figure}
    \centering
    \includegraphics{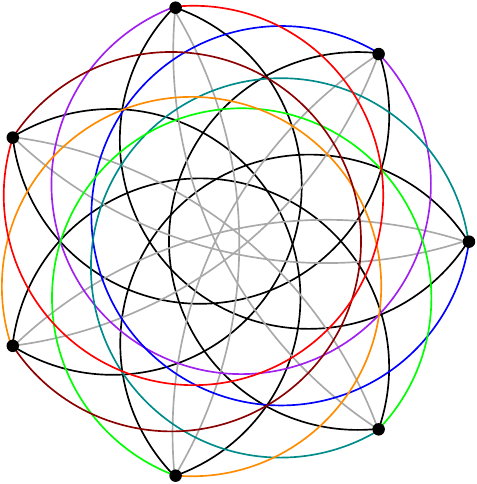}
    \caption{Any two adjacent edges cross exactly once, any two non-adjacent edges cross at least once and at most twice.}
    \label{fig:flower}
\end{figure}

In general, a \textit{graph} in the abstract sense consists of a vertex set and an edge set, where each edge is a vertex pair (unordered). A \textit{drawing} of a graph is a representation of its vertices as distinct points in the plane, and its edges as curves connecting the corresponding pair of vertices. We make the same mild assumptions for all drawings of graphs. A drawing is said to be \textit{plane} if it has no crossings. Planarity has been a main theme in the development of graph theory featured by famous theorems such as Euler's formula and the four color theorem. A drawing is said to be \textit{simple} if it is adjacent-simple and separate-simple simultaneously. Rafla~\cite{rafla1988} conjectured that every simple complete graph drawing on at least three vertices must contain a plane Hamiltonian cycle, that is, a cycle that visits every vertex exactly once. Motivated by this conjecture or not, unavoidable plane structures in simple drawings are extensively studied, see our incomplete list~\cite{aichholzer2022shooting,aichholzer2024twisted,fox2025structure,fulek14disjoint,fulek2013topological,garcia2021plane,pach2010disjoint,ruizvargas2017disjoint,suk2013disjoint,suk2025unavoidable,zeng2025note}. We emphasize the work~\cite{pach2010disjoint} of Pach and T\'oth, in which they proved simple complete graph drawings induce many pairwise disjoint edges and constructed complete graph drawings where any two edges cross at least once and at most twice. However, some adjacent edges cross twice in their construction. With respect to the conditions, Theorem~\ref{thm:main1} can be considered as an improvement of their result (albeit quantitatively they obtain more disjoint edges). Further work on unavoidable plane structures also addresses simple complete graph drawings that satisfy additional conditions~\cite{aichholzer2024separable,bergold2025plane,bergold2024holes}, and especially those where the edges have to be straight~\cite{pach1994some,sharir2013counting,toth2000note}. Drawings that are only adjacent-simple or only separate-simple, have been studied under the names semi-simple~\cite{balko2015semi} and star-simple~\cite{felsner2020maximum}, or tolerable~\cite{cairns2019bad}, respectively. This paper attempts to capture the earliest moment at which one can expect unavoidable plane structures.

We shall state the two cases of Theorem~\ref{thm:main1} as theorems on their own. We call a graph $S$ a \textit{plane structure guaranteed} in a collection $\mathcal{C}$ of graph drawings if there exists an integer $n$ such that any drawing in $\mathcal{C}$ of at least $n$ vertices must contain a plane sub-drawing of $S$. Here, a drawing $A$ is called a sub-drawing of another drawing $B$ if the vertex set and edge set of $A$ are subsets of those of $B$ respectively. A \textit{squid} is a connected graph that contains exactly one cycle which is a triangle, and all other vertices only form edges with two fixed vertices on this triangle. A \textit{caterpillar} is a connected graph without cycles that contains a central path, and all other vertices only form edges with vertices on this path. An \textit{isolated vertex} in a graph is a vertex that does not form any edges.
\begin{figure}[ht]
    \centering
    \includegraphics[page=1]{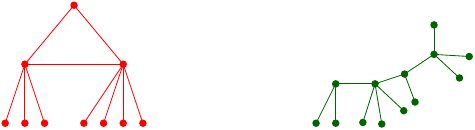}
    \caption{A squid (left) and a caterpillar (right).}
    \label{fig:catersquid}
\end{figure}

\begin{theorem}\label{thm:main2}
    The plane structures guaranteed in adjacent-simple complete graph drawings are  exactly squids with extra isolated vertices, and unions of disjoint caterpillars.
\end{theorem}

\begin{theorem}\label{thm:main3}
    The plane structures guaranteed in separate-simple complete graph drawings are  exactly unions of disjoint edges and isolated vertices.
\end{theorem}

One may ask to what extent our main result still holds if the mild assumptions are dropped. For example, when we connect points on a piece of paper, the curves drawn by a pen could touch or run through any points, then what would be the earliest cause of disjointness in this scenario? To model an edge in a drawing without mild assumptions, we use the image $\text{Im}(f)$ of a function $f$ from the unit interval $[0,1]$ to the plane $\mathbb{R}^2$. We also define self-intersection of $f$ as $\Delta(f)=\{p: p=f(x)=f(y)~\text{for $x\neq y$}\}$. A family of functions $\{f_i\}$ is called \textit{stroke-like} if all the pre-images $f_i^{-1}(\text{Im}(f_j))$ and $f_i^{-1}(\Delta(f_i))$ are unions of finitely many sub-intervals of $[0,1]$. Although mathematically two curves can intersect infinitely many times without coincide along any segment, it is not physically possible to draw non-stroke-like curves on a piece of paper using a pen. We can conclude the following corollary of Theorem~\ref{thm:main1}.
\begin{corollary}\label{cor:main0}
    For any positive integer $m$, there exists $n$ such that the following is true. Let $p_1,p_2,\dots,p_n$ be distinct points in the plane, $c_{ij}: [0,1] \to \mathbb{R}^2$ be piecewise smooth function with $c_{ij}(0)=p_i$ and $c_{ij}(1)=p_j$ for all $1\leq i<j \leq n$, and suppose the family $\{c_{ij}\}$ is stroke-like. If either of the following conditions holds: \begin{enumerate}
        \item[(A)] $\forall$ $i<j$ and $k<\ell$ with $|\{i,j,k,\ell\}| = 3$, we have $|c_{ij}^{-1}(\text{Im}(c_{k\ell}))| = |c_{k\ell}^{-1}(\text{Im}(c_{ij}))| = 1$;
        \item[(S)] $\forall$ $i<j$ and $k<\ell$ with $|\{i,j,k,\ell\}| = 4$, we have $|c_{ij}^{-1}(\text{Im}(c_{k\ell}))| = |c_{k\ell}^{-1}(\text{Im}(c_{ij}))| \leq 1$, and there is no touching point between $c_{ij}$ and $c_{k\ell}$;
    \end{enumerate} then there exist $m$ functions in $\{c_{ij}:~1\leq i<j \leq n\}$ whose images are pairwise disjoint.
\end{corollary}

We will discuss the necessity of Conditions (A) and (S) towards the end of this paper.

\section{The adjacent-simple case}\label{sec:adjacent}

We shall prove Theorem~\ref{thm:main2} in this section. The starting point is the observation below.
\begin{proposition}\label{prop:adjacent}
    There exists at least one pair of disjoint edges in every adjacent-simple complete graph drawing on four vertices.
\end{proposition}

\begin{proof}
We argue with an adjacent-simple complete graph drawing whose vertices are $a,b,c,d$. Each edge in this drawing crosses at most one edge (the others are adjacent). If $ab$ and $cd$ do not cross, they form the desired pair of disjoint edges. Hence we assume $ab$ and $cd$ cross, and argue that some other pair of edges must be disjoint. For an edge $uv$, we call a crossing point $p$ the \textit{first crossing of $uv$ as seen from $u$} if the edge segment of $uv$ from $u$ to $p$ is not crossed by any edges.

Let $p$ be the first crossing of $cd$ as seen from $c$. By adjacent-simplicity, $p$ lies on $ab$. Consider the simple closed curve that is the concatenation of the edge $ac$, the segment of $cd$ from $c$ to $p$, and the segment of $ab$ from $p$ to $a$. By the Jordan curve theorem, it cuts the plane into two regions, and we refer to the region not containing $b$ as $R$. We assume $ac$ and $bd$ cross, otherwise we are done. Let $q$ be the first crossing of $bd$ as seen from $b$. Consider the simple closed curve that is the concatenation of the edge $ab$, the segment of $bd$ from $b$ to~$q$, and the segment of $ac$ from $q$ to $a$. Again this curve cuts the plane into two regions, and we refer the region containing $c$ as $S$. Note that the boundaries between $R$ and $S$ do not cross (although they are partly identical and partly disjoint). We conclude that $R \subset S$. We can assume $S$ is bounded after a projective transformation (hence $R$ is also bounded), see Figure~\ref{fig:K4} for an illustration.

Now we consider the edge $ad$ as a curve emanating from $a$. Due to adjacent-simplicity, $ad$ cannot cross the boundaries of $R$ nor $S$. If $ad$ emanates from $a$ at a direction inside $R$, then $ad$ lies fully inside $R$. If $ad$ emanates from $a$ at a direction outside $R$, then $ad$ lies fully outside $S$. On the other hand, the edge $bc$ cannot cross the boundary of $S\setminus R$ by adjacent-simplicity. This implies $bc$ lies fully inside $S \setminus R$. Thus, $ad$ and $bc$ form a pair of disjoint edges.
\end{proof}
\begin{figure}[ht]
		\centering
		\includegraphics[page=1]{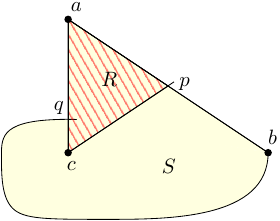}
		\caption{$R$ (marked with red stripes) is fully contained in $S$ (shaded in yellow).}
        \label{fig:K4}
\end{figure}

We call an adjacent-simple complete graph drawing \textit{Type~I}, \textit{Type~II}, or \textit{Type~III} if we can order its vertices as $v_1,v_2,\dots,v_n$ such that conditions I, II, or III below are satisfied respectively:
\begin{enumerate}
    \item[I.] $v_i,v_j,v_k,v_\ell$ induce only one pair of disjoint edges $v_iv_j$ and $v_kv_\ell$ for all $i<j<k<\ell$.

    \item[II.] $v_i,v_j,v_k,v_\ell$ induce only one pair of disjoint edges $v_iv_k$ and $v_jv_\ell$ for all $i<j<k<\ell$.

    \item[III.] $v_i,v_j,v_k,v_\ell$ induce only one pair of disjoint edges $v_iv_\ell$ and $v_jv_k$ for all $i<j<k<\ell$.
\end{enumerate}

See Figures~\ref{fig:typeI}, \ref{fig:typeII}, and~\ref{fig:typeIII} for illustrations of type-drawings of six vertices. It is not trivial to find these figures without a priori knowledge of their existence.

\begin{figure}
    \centering
    \includegraphics[page=1,scale=1.1]{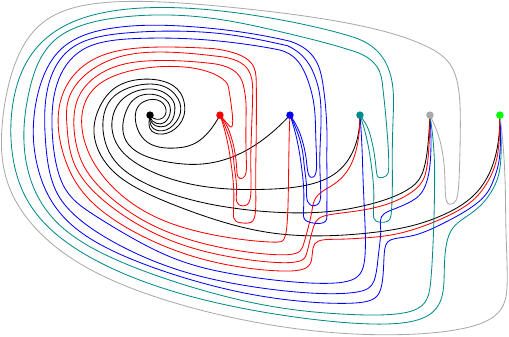}
    \caption{A Type~I drawing of vertices $v_1,v_2,\dots,v_6$ from left to right.}\label{fig:typeI}

    \bigskip
    \includegraphics[page=2,scale=1.1]{type-drawings.pdf}
    \caption{A Type~II drawing of vertices $v_1,v_2,\dots,v_6$ from left to right.}\label{fig:typeII}

    \bigskip
    \includegraphics[page=3,scale=1.1]{type-drawings.pdf}
    \caption{A Type~III drawing of vertices $v_1,v_2,\dots,v_6$ from left to right.}\label{fig:typeIII}
\end{figure}

\begin{proposition}\label{prop:type}
    There exist adjacent-simple complete graph drawings of Type~I, Type~II, and Type~III with arbitrarily many vertices.
\end{proposition}
\begin{proof}[Proof Sketch]
    One approach towards a proof is to extend the drawing pattern in Figures~\ref{fig:typeI}, \ref{fig:typeII}, and~\ref{fig:typeIII} to arbitrarily many vertices. We briefly explain the idea on how to construct the drawing in Figure~\ref{fig:typeIII}. The patterns depicted in Figures~\ref{fig:typeI} and~\ref{fig:typeII} are observable to a moderate level, hence we skip their explanations (but curious readers are encouraged to try).

    To obtain a Type~III drawing depicted in Figure~\ref{fig:typeIII}, we first put vertices $v_1,v_2,\dots,v_n$ on a horizontal line such that vertex $v_i$ is left of vertex $v_j$ for any $i<j$. Then we inductively consider each $v_i$ and draw the edges from $v_i$ to $v_j$ for all $j > i$. When $i=1$ or $2$, the edges are created as upper-half circular arcs between $v_1$ and $v_j$. When $i > 2$, the edges $v_iv_j$ are drawn in the following way: we create a curve emanating from $v_i$ that crosses every edge $v_av_b$ with $a<b<i$ but does not cross edges incident to $v_i$; we do not require this curve to end at a vertex so its existence can be guaranteed; next, we turn this curve into $n-i$ thin loops nesting each other; then we slightly move the endpoints of these loops to the right of $v_i$ at different distances; we further create upper-half circular arcs between these endpoints and vertices $v_j$ for all $j > i$; we make sure endpoints closer to $v_i$ are connected to further vertices; finally, the edges $v_iv_j$ are the concatenations of the near-loops, starting at $v_i$ and ending near it, together with the circular arcs, starting near $v_i$ and ending at $v_j$.
\end{proof}

Our final ingredient for Theorem~\ref{thm:main2} is the celebrated theorem of Ramsey~\cite{ramsey1930problem}.
\begin{theorem}\label{thm:ramsey}
    For any positive integers $k, m, r$, there exists an integer $n$ with the following property. For every partition of all size-$r$ subsets of a size-$n$ set $S$ into $k$ classes, there exists a size-$m$ subset $S'$ of $S$ such that all size-$r$ subsets of $S'$ belong to the same class in this partition.
\end{theorem}

\begin{proof}[Proof of Theorem~\ref{thm:main2}]

    Let $n$ be guaranteed by Theorem~\ref{thm:ramsey} with $k=3$, an arbitrarily fixed $m$, and $r=4$. For any adjacent-simple complete graph drawings of vertices $v_1,v_2,\dots,v_n$, we can partition its quadruples $v_i,v_j,v_k,v_\ell$ with $i<j<k<\ell$ into three classes as follows: if $v_iv_j$ and $v_kv_\ell$ are disjoint, we put this quadruple in class~I; else if $v_iv_k$ and $v_jv_\ell$ are disjoint, we put this quadruple in class~II; else $v_iv_\ell$ and $v_jv_k$ must be disjoint by Proposition~\ref{prop:adjacent}, and we put this quadruple in class~III. According to Theorem~\ref{thm:ramsey}, there exist $m$ vertices such that all quadruples among them belong to the same class. Depending on whether this is class~I, II, or~III, any two edges induced by them must be disjoint if they are disjoint in a drawing of Type~I, II, or~III respectively (with the same vertex order). Thus, a graph $G$ is a plane structure if all type-drawings of $m$ vertices contain a plane sub-drawing of $G$. For the remainder of the proof, we only consider type-drawings, and by a mild abuse of notation, we denote their vertices as $v_1,v_2,\dots,v_m$ in order. To find a plane sub-drawing of $G$ inside a type-drawing, it suffices to arrange its vertices into a sequence, and identify them with $v_1,\dots,v_m$ according to the order of this sequence, then we can check that crossings are not induced by the corresponding edges under such an identification.

    Let $S$ be an arbitrary squid within which $a,b,c$ form a triangle, $u_1,\dots,u_s$ are only adjacent to $a$, and $w_1,\dots,w_t$ are only adjacent to $c$. We shall arrange these vertices into a sequence based on the type-drawing under consideration. In a drawing of Type~I, the sequence can be $u_1,\dots,u_s$, $a$, $b$, $c$, $w_1,\dots,w_t$. In a drawing of Type~II, the sequence can be $a$, $w_1,\dots,w_t$, $b$, $u_1,\dots,u_s$, $c$. In a drawing of Type~III, the sequence can be $a$, $w_1,\dots,w_t$, $c$, $b$, $u_1,\dots,u_s$. It is easy to verify that the edges do not cross in the correspondingly induced sub-drawings of $S$. Also, it is very easy to arrange isolated vertices anywhere in these sequences without inducing crossings. Hence, we conclude squids with extra isolated vertices are plane structures as desired.

    Let $U$ be an arbitrary caterpillar union whose components are denoted by $C_1,C_2,\dots,C_s$ where each $C_i$ is a caterpillar. We shall arrange a vertex sequence for each $C_i$ and then combine them into a sequence for $U$. For this purpose, we use $C$ to denote an arbitrary caterpillar, $u_1,u_2,\dots,u_t$ to denote its central path, and $w^i_j$ to denote its non-path vertices adjacent to $u_i$. For any fixed $i$, the relative order among $w^i_j$ is arbitrary in our sequences and the range of $j$ will be unimportant.
    
    Inside a Type~I drawing, we arrange the vertices of $C$ as follows: $u_1,u_2,\dots,u_t$ appear increasingly, in other words, $u_i$ is before $u_j$ for $i<j$; $w^i_j$ appears between $u_i$ and $u_{i+1}$ for $i<t$; $w^t_j$ appears after $u_t$. We denote such a sequence for $C$ as $S(C)$. Then, the sequence for $U$ can be $S(C_1)$, $S(C_2)$, $\dots$, $S(C_s)$. If we identify the vertices of $U$ with those of a Type~I drawing according to this sequence, we can check that any two non-adjacent edges of $U$ are identified as $v_iv_j$ and $v_kv_\ell$ with $i<j<k<\ell$, hence are disjoint by the definition of Type~I drawings. Since type-drawings are adjacent-simple, the induced sub-drawing of $U$ is plane.

    Inside a Type~II drawing, we arrange the vertices of $C$ into two sequences $L(C)$ and $R(C)$: $L(C)$ consists of all $u_i$ with odd $i$ appearing increasingly, followed by all $w^i_j$ with even $i$ appearing increasingly (with respect to $i$); $R(C)$ consists of all $w^i_j$ with odd $i$ appearing increasingly (with respect to $i$), followed by all $u_i$ with even $i$ appearing increasingly. Then, the sequence for $U$ can be $L(C_1)$, $L(C_2)$, $\dots$, $L(C_s)$, $R(C_1)$, $R(C_2)$, $\dots$, $R(C_s)$. Similarly, using the construction of this sequence and the definition of Type~II drawings, we can verify that this sequence induces a plane sub-drawing of $U$.
    
    Inside a Type~III drawing, we first arrange the central path in a zig-zag fashion: $u_1$, $u_3$, $u_5,\dots,u_{t-1}$, $u_{t}$, $u_{t-2},\dots,u_4$, $u_2$ if $t$ is even; $u_1$, $u_3$, $u_5,\dots,u_{t-2}$, $u_{t}$, $u_{t-1},\dots,u_4$, $u_2$ if $t$ is odd. On top of this, $w^1_j$ appears after $u_2$; $w^i_j$ appears between $u_{i-1}$ and $u_{i+1}$ for $1<i<t$; $w^t_j$ appears between $u_{t-1}$ and $u_t$. This describes a sequence $S(C)$ for $C$, and we break it into two sub-sequences: $L(C)$ is the prefix of $S(C)$ before but not including $u_t$ if $t$ is even; $L(C)$ is the prefix of $S(C)$ until and including $u_t$ if $t$ is odd; $R(C)$ is the suffix of $S(C)$ after deleting $L(C)$. Then, the sequence for $U$ can be $L(C_1)$, $L(C_2)$, $\dots$, $L(C_s)$, $R(C_s)$, $\dots$, $R(C_2)$, $R(C_1)$. Again, we can check that the induced sub-drawing of $U$ is plane.

    It suffices for us to show that a plane structure $G$ must be either a squid with extra vertices or a union of disjoint caterpillars. According to Proposition~\ref{prop:type}, all sufficiently large type-drawings must contain a plane sub-drawing of $G$. We shall exploit this fact to restrict the structure of $G$.

    First, we argue that $G$ cannot contain a cycle of length more than three. Suppose $G$ contains a cycle of length $s \geq 4$, we let $u_1,u_2,\dots,u_s$ denote the vertices appearing along this cycle. Consider a plane sub-drawing of $G$ inside a Type~I drawing, and following the vertex order of this type-drawing, $u_1,u_2,\dots,u_s$ appear in a sequence. Without loss of generality, we assume $u_2$ appears after $u_1$. Then, $u_3$ cannot appear between $u_1$ and $u_2$, otherwise $u_1u_2$ crosses $u_3u_4$ (independent of the position of $u_4$). On the other hand, $u_3$ cannot be before $u_1$ otherwise $u_1u_s$ crosses $u_2u_3$. Hence $u_3$ must appear after $u_2$. By the same argument, $u_4$ must be after $u_3$, and in general $u_i$ must be after $u_{i-1}$. However, by the cyclic nature and the same argument, we can reach the conclusion that $u_1$ appears after $u_s$, which is a contradiction.

    Next, we argue that $G$ must be a squid if it contains a triangle but no isolated vertices. Let $a,b,c$ be the vertices of this triangle. Consider a plane sub-drawing of $G$ inside a Type~II drawing. Without loss of generality, we assume $a$ is before $b$ and $b$ is before $c$ in the vertex order of this type-drawing. Then any edge of $G$ not incident to $a$, $b$, or $c$ must have a vertex between $a$ and $b$ (to not cross $ab$) and a vertex between $b$ and $c$ (to not cross $bc$). However, such an edge would cross $ac$, thus no such edge exists. Together with the fact that $G$ cannot contain a cycle of length four, we conclude any vertex $x \neq a,b,c$ is incident to exactly one edge, formed with $a$, $b$, or $c$. Now, suppose $a,b,c$ all form edges with non-triangle vertices. Consider a plane sub-drawing of $G$ inside a Type~I drawing, and assume $b$ is between $a$ and $c$ in the vertex order of this type-drawing. Then an edge $bx$ with $x\neq a,c$ must cross $ac$, which is a contradiction. Hence, one of $a,b,c$ does not form edges outside the triangle, and we conclude $G$ must be a squid.

    \begin{figure}[ht]
        \centering
        \includegraphics[page=2]{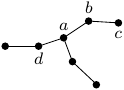}
        \caption{A graph contained in all acyclic graphs except unions of disjoint caterpillars.}
        \label{fig:caterpillarkiller}
    \end{figure}

    Lastly, we argue that $G$ must be a caterpillar union if it is acyclic. It is sufficient to show that $G$ does not contain the sub-graph depicted in Figure~\ref{fig:caterpillarkiller}. Suppose $G$ contains this sub-graph and some of its vertices are $a,b,c,d$ as in Figure~\ref{fig:caterpillarkiller}. Consider a plane sub-drawing of $G$ inside a Type~I drawing. By the pigeonhole principle, two neighbors of $a$ are on the same side of $a$ (before or after) in the vertex order of this type-drawing. Without loss of generality, we assume that $b$ and $d$ are on the same side of $a$, and $b$ is closer to $a$ than $d$. Then $b$ appears between $a$ and $d$, making $bc$ cross $ad$ (independent of the position of $c$), which is a contradiction. This concludes the proof.
\end{proof}

\section{The separate-simple case}\label{sec:separate}

We shall prove Theorem~\ref{thm:main3} in this section, starting with the observation below. A \textit{cell} of a drawing (not necessarily plane) refers to a connected component of the complement of this drawing inside the plane.

\begin{lemma}\label{lem:samecell}
    In any separate-simple drawing of a triangle $abc$ union an edge $de$, such that the vertices $d,e$ are in the same cell of $abc$, the edge $de$ does not cross one of~$ab,bc,ac$.
\end{lemma}

\begin{proof}
Two cells of the triangle $abc$ is said to be adjacent if their boundaries coincide along a segment of an edge. We can assign black and white colors to the cells of $abc$ such that adjacent cells have different colors. Such a ``chessboard'' coloring can be argued with standard graph theory knowledge, see e.g. Section~4.6 of~\cite{diestel2025graph} (especially Exercise~24), or we may simply apply Proposition~2 in \cite{hertrich2022coloring}. Now, as we travel along the trajectory of the edge $de$, the cell color of our location changes exactly when we encounter a crossing between $de$ and edges of $abc$. Since $d,e$ are in the same cell of $abc$, the number of such crossings must be even. Hence $de$ cannot cross all three of~$ab,bc,ac$ by separate-simplicity. We remark that multiple edges are allowed to cross at the same point, if $de$ passes through a crossing point between $ab$, $ac$, or $bc$, the no-touching mild assumption restricts the behavior of $de$ and our proof still works.
\end{proof}

\begin{proposition}\label{prop:separate}
    There exists at least one pair of disjoint edges in every separate-simple complete graph drawing on forty-three vertices. 
\end{proposition}

\begin{proof}
We prove that $r = 18 + 18 + 1 + 6 = 43$ vertices suffice to induce a pair of disjoint edges. We refer to these vertices as $a,b,c,d,e,f$ and \textit{the rest}. We shall argue via the cells of the triangle $abc$ and the triangle $def$. If the union of a cell $\Omega$ of $abc$ and another cell $\Omega'$ of $def$ is the whole plane, then, by the pigeonhole principle, either $\Omega$ or $\Omega'$ contains at least two vertices among the rest. In this case, the edge between those two vertices is disjoint from one of $ab,ac,bc$ or $de,df,ef$ by Lemma~\ref{lem:samecell}. Therefore, we assume for the remainder of the proof that this case does not happen.

Under this assumption, we investigate how cells of $abc$ and $def$ can interact. We call a cell of $abc$ \textit{covered} if it is contained in a cell of $def$. We call a cell of $def$ \textit{covering} if it contains a cell of $abc$. A cell of $abc$ is said to be \textit{effective} if it contains a vertex from the rest. If there is a cell of $abc$ that contains two vertices from the rest, then we are done by Lemma~\ref{lem:samecell}. Otherwise, there must be $18+18+1$ effective cells of $abc$. We will argue that there are at most $18$ non-covered cells of $abc$ and at most $18$ covering cells of $def$. Thus, there must be two effective cells of $abc$ contained in the same cell of $def$ by the pigeonhole principle. Consequently, the vertices from the rest in those two effective cells are contained in the same cell of $def$, and again, the proof can be concluded by Lemma~\ref{lem:samecell}.

It suffices for us to prove the bound on non-covered cells and covering cells. To that end, we need the following claim, whose proof we postpone. A cell $\Omega$ of $abc$ is said to \textit{overlap} a cell $\Omega'$ of $def$ if $\Omega \cap \Omega' \neq \emptyset$, $\Omega \not \subset \Omega'$, and $\Omega' \not \subset \Omega$.

\noindent \textbf{Claim.} If a bounded cell $\Omega$ of $abc$ overlaps a bounded cell $\Omega'$ of $def$, there exists a point on the boundaries of $\Omega$ and $\Omega'$ that is a crossing between an edge among $ab,ac,bc$ and another edge among $de,df,ef$.

Note that the claim still holds when $\Omega$ and $\Omega'$ are not necessarily bounded. Indeed, by our assumption at the beginning of the proof, there is a point $p$ not in $\Omega \cup \Omega'$, thus we could apply a projective transformation to locate $p$ at infinity, so we can conclude the existence of such a crossing using the original claim.

To bound the number of non-covered cells of $abc$ and covering cells of $def$, we show that each such cell must overlap some cell of the other triangle. Let $\Omega_1$ be a non-covered cell of~$abc$. We take a point $p$ on the boundary of $\Omega_1$ and not on the edges of $def$. Let $\Omega_1'$ be the cell of $def$ that contains a small neighborhood of $p$, then we can argue $\Omega_1$ overlaps $\Omega_1'$. Let $\Omega_2'$ be a covering cell of $def$. Consider the union $U$ of all cells of $abc$ contained in $\Omega_2'$. Let $\Omega_2$ be a cell of $abc$ containing a point $p \in \Omega_2' \setminus U$, then we can argue $\Omega_2$ overlaps $\Omega_2'$. According to the claim, $\Omega_1$ and $\Omega_2'$ must be incident to some crossing point between an edge of $abc$ and an edge of $def$. By separate-simplicity, there are at most $9$ such crossing points, and each of them is incident to at most two cells of $abc$. Hence, there are at most $18$ non-covered cells of $abc$ and at most $18$ covering cells of $def$.

\begin{figure}[ht]
    \centering
    \includegraphics[page=2]{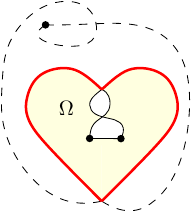}
    \caption{A bounded cell $\Omega$ (shaded) of a triangle drawing, whose contour is drawn in bold and red. Segments not on the boundary of $\Omega$ are drawn dashed.}
    \label{fig:contour_example}
\end{figure}

\noindent \textbf{Proof of Claim.} We define the \textit{contour} of $\Omega$ to be the set of all points on the boundary of $\Omega$ from which one can draw a curve to the unbounded cell without crossing any other point on the boundary; see Figure~\ref{fig:contour_example} for an example. We argue that the contour of $\Omega$ must be a simple closed curve. As $\Omega$ is a bounded cell, its contour must contain a loop (to prevent the points in the interior of $\Omega$ from being in the unbounded cell). Let $L$ be such a loop. Any segment on the boundary of $\Omega$ that is in the interior of $L$ cannot be part of the contour by definition (as the loop blocks the way to the unbounded cell). As a bounded cell is connected, there cannot be parts of $\Omega$ outside $L$ and thus any segments outside of $\Omega$ cannot be on the boundary of~$\Omega$. This means the contour of $\Omega$ is $L$, which is a simple closed curve. By the Jordan curve theorem, the contour of $\Omega$ cuts the plane into two regions, and the bounded one contains $\Omega$.

We also argue that the boundary of $\Omega$ must be connected. To this end, we show that the unbounded cell must have a connected boundary. It follows that the boundary of a cell must be connected, because a projective transformation can turn any bounded cell into the unbounded cell without changing the connectedness of the boundary. Assume, for a contradiction, that the boundary of the unbounded cell consists of several connected components. As the cells are induced by three adjacent edges (forming a potentially self-crossing cycle), there must be segments of those edges connecting the components of this boundary. For an arc not incident to the unbounded cell, it must be inside a bounded cell and thus enclosed by a loop. However, being enclosed by a loop would prevent it from connecting different boundary components of the unbounded cell.

The contour of $\Omega'$ is defined analogously and has the same properties as the contour of $\Omega$. If the contour of $\Omega$ and the contour of $\Omega$ intersect, their intersection point is a crossing point $p$ on their boundaries. If these two contours do not intersect, they must be nested, that is, one of the contours must be outside the other, wrapping around it because $\Omega$ overlaps $\Omega'$. Assume, without loss of generality, the contour of $\Omega$ wraps around outside the contour of $\Omega'$. Then the contour of $\Omega$ is not contained inside $\Omega'$. We find a point $q$ on the boundary of $\Omega$ inside $\Omega'$ in the following way: we start with a point $q' \in \Omega' \setminus \Omega$ and traverse to a point $q'' \in \Omega \cap \Omega'$ along a trajectory that stays completely in $\Omega'$; as $q'$ is not in $\Omega$ and $q''$ is in $\Omega$, the boundary must have been encountered in a point $q$. Traversing along the boundary of $\Omega$ from $q$ to some point on the contour of $\Omega$, the boundary of $\Omega'$ must have been crossed (in order to get from inside to outside $\Omega'$). Thus, the boundaries of $\Omega$ and $\Omega'$ cross.
\end{proof}

\begin{proof}[Proof of Theorem~\ref{thm:main3}]
    We name the vertices as $v_1,v_2,\dots,v_n$ and order them such that $v_i < v_j$ for $i<j$. 
    According to Proposition~\ref{prop:separate}, among any $r=43$ vertices $u_1 < u_2 < \dots < u_r$, there exists a pair of disjoint edges $u_iu_j$ and $u_ku_\ell$ with $i<j$, $k<\ell$, and $i<k$. We put this $r$-tuple into the class labeled by $(i,j,k,\ell)$, hence obtain a partition of all $r$-tuples of vertices into classes. By Theorem~\ref{thm:ramsey} and assuming $n$ to be sufficiently large, we can find $(2m+1)r$ many special vertices $w_1<w_2<\dots<w_{(2m+1)r}$ such that every $r$-tuple among them are in the same class. By an abuse of notation, we use $(i,j,k,\ell)$ to denote the label of this class.

    If $i<j<k<\ell$, we argue the edges $w_{(2s-1)r}w_{2sr}$ $(1\leq s\leq m)$ are pairwise disjoint. Indeed, for any $s<t$, we can choose as an $r$-tuple of the vertices such that $w_{(2s-1)r}$, $w_{2sr}$, $w_{(2t-1)r}$, $w_{2tr}$ are the $i$-th, $j$-th, $k$-th, $\ell$-th element in this $r$-tuple respectively. This is possible due to the many special vertices that are in between $w_{(2s-1)r}$, $w_{2sr}$, $w_{(2t-1)r}$, and~$w_{2tr}$. Since this $r$-tuple is in the class labeled by $(i,j,k,\ell)$, the edges $w_{(2s-1)r}w_{2sr}$ and $w_{(2t-1)r}w_{2tr}$ are disjoint. If $i<k<j<\ell$, we can argue the edges $u_{sr}u_{(s+m)r}$ $(1\leq s\leq m)$ are pairwise disjoint analogously. If $i<k<\ell<j$, we can argue the edges $u_{sr}u_{(2m+1-s)r}$ $(1\leq s\leq m)$ are pairwise disjoint analogously. Therefore, in any case, there are $m$ pairwise disjoint edges. Because the number $m$ is arbitrary and isolated vertices can be dealt with easily, we conclude that unions of disjoint edges and isolated vertices are plane structures.

    \begin{figure}[ht]
    \centering
    \includegraphics[page=7]{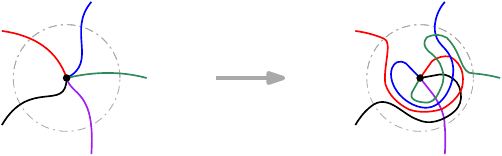}
    \caption{Redrawing the edges around a vertex to make them cross each other.}
    \label{fig:separate_simple_adjacent_cross}
    \end{figure}

    Lastly, we show that graphs other than unions of disjoint edges and isolated vertices cannot be the plane structures guaranteed in separate-simple complete graph drawings. Indeed, it is easy to construct separate-simple complete graph drawings in which any two adjacent edges cross: fix an arbitrary simple complete graph drawing; redraw the initial parts of all edges incident to each vertex in such a way that the edges now emanate in the opposite cyclical order at this vertex and then pairwise cross; the rest of the drawing remains unchanged, so the resulting drawing is separate-simple. See Figure~\ref{fig:separate_simple_adjacent_cross} for a depiction.
\end{proof}

\section{The main result}\label{sec:main}

\begin{proof}[Proof of Theorem~\ref{thm:main1}]
Note that a union of disjoint edges is also a union of disjoint caterpillars. The first part of this theorem follows from combining Theorem~\ref{thm:main2} and Theorem~\ref{thm:main3}.

For the second part of this theorem, we describe a construction for $n$ vertices; see Figure~\ref{fig:flower}. Write $d = \lfloor n/2 \rfloor$ and choose real numbers $1 > r_1 > r_2 > \dots > r_d > 0$. Fix a circle with radius one centered at a point $o$. On this circle, we pick $n$ points with equal gaps, which will serve as vertices. There are exactly $d$ distinct distances between these vertices. We say a pair of vertices is $i$-distant if their distance is the $i$-th smallest among all distinct distances. For any $i$-distant pair $a$ and $b$, let $L$ be the line passing through $o$ and perpendicular to the line $Z$ determined by $a$ and $b$. There are two points on $L$ whose distance to $o$ equals~$r_i$. Between these two points, we choose the point $c$ whose distance to $Z$ is larger (breaking ties arbitrarily). The edge between $a$ and $b$ is drawn as the arc on the circle determined by $a,b$, and $c$ which passes through $c$. One can easily see that in this way, any two adjacent edges cross exactly once, and any two non-adjacent edges cross at least once and at most twice.
\end{proof}

\begin{proof}[Proof of Corollary~\ref{cor:main0}]
By abuse of notation, we shall call each $p_i$ a vertex and each $c_{ij}$ (or its image) an edge. We shall perform local modifications to the given configuration of vertices and edges, and obtain a drawing with mild assumptions that is either adjacent-simple or separate-simple. We make sure not to introduce any new pair of disjoint edges, so that Theorem~\ref{thm:main1} can be applied to conclude the proof.

Condition (A) means any two adjacent edges have exactly one intersection counted via pre-images. In particular, an edge $ab$ cannot pass through another vertex $c$ (otherwise $ac$ and $ab$ have two intersections), and an edge $ab$ cannot pass through $a$ or $b$ in its interior (otherwise another edge incident to $a$ or $b$ intersects $ab$ twice counted via pre-images). We now eliminate degeneracies without introducing disjointness or violating adjacent-simplicity.

\begin{figure}[ht]
    \centering
    \includegraphics[page=5]{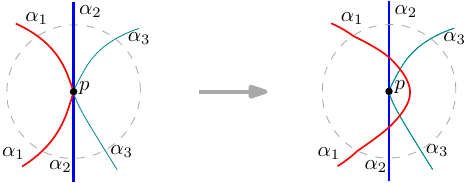}
    \caption{We do not perturb a segment (such as $\alpha_2$) sandwiched by others.}
    \label{fig:perturb_touch}
\end{figure}

\begin{figure}[ht]
    \centering
    \includegraphics[page=9]{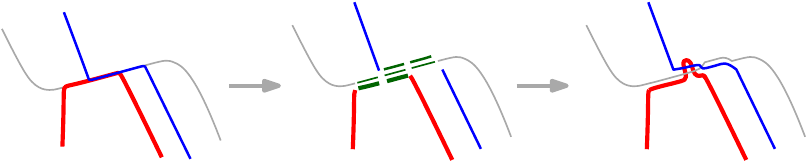}
    \caption{A strand is locally twisted to avoid introducing disjointness.}
    \label{fig:unpack}
\end{figure}

\begin{figure}[ht]
    \centering
    \includegraphics[page=8]{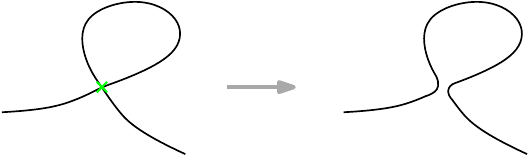}
    \caption{A standard local redraw to resolve self-crossings.}
    \label{fig:selfcross}
\end{figure}

\begin{figure}[ht]
    \centering
    \includegraphics[page=6]{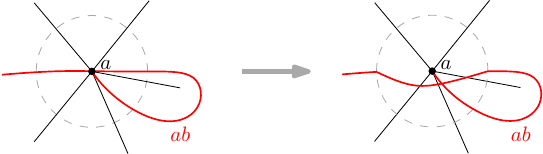}
    \caption{Any two adjacent edges intersect before and after the perturbation.}
    \label{fig:perturb_block}
\end{figure}

First, if there is a touching point $p$ (whether a self-touching or between distinct edges), we locally perturb the drawing near $p$ to decrease the total number of touchings (counted via pre-images). A segment $\alpha_2$ is said to be \textit{sandwiched} if there are segments $\alpha_1$ and $\alpha_3$ that touch it in $p$ from two sides. To avoid introducing disjointness, we always perturb a segment through $p$ that is not sandwiched there, see Figure~\ref{fig:perturb_touch}. The existence of such a segment is justified at the end of the proof. Since touchings only happen between non-adjacent edges under Condition (A), the new crossings introduced in this step will not violate adjacent-simplicity.

Next, we remove coincidences between edges along segments. For any maximal segment $\sigma$ along which edges coincide, its multiplicity is defined as the total number of pre-image sub-intervals mapping onto $\sigma$ among all the functions $c_{ij}$. We unpack $\sigma$ into that many duplicated parallel strands lying very close to each other. Since coincidences along segments do not happen between adjacent edges under condition (A), we may reconnect the resulting parallel strands arbitrarily (in a local manner) to form edges. Moreover, we can locally twist some strands to avoid introducing disjointness, see Figure~\ref{fig:unpack}.

Then, we resolve crossing points that are self-intersections by a standard redraw technique, see Figure~\ref{fig:selfcross}. Although our mild assumption does not forbid multiple segments crossing at the same point, we can assume such a scenario does not happen after slight perturbations, which makes it more intuitive to resolve self-crossings. At the end, we obtain an adjacent-simple complete graph drawing satisfying the mild assumptions as wanted.

Condition (S) means any two non-adjacent edges have at most one intersection counted via pre-images, and that no such intersection is a touching point. Under this condition, an edge $ab$ cannot pass through two other vertices $c$ and $d$, otherwise $ab$ and $cd$ have two intersections. In particular, an edge $ab$ cannot pass through two other vertices $c$ and $d$, since then $ab$ and $cd$ would meet at least twice (at $c$ and $d$). Therefore, for $n$ sufficiently large we may greedily select as many vertices as we wish so that, among the selected vertices, no induced edge contains other selected vertices in its interior. We restrict our considerations to these selected vertices. Moreover, we can use local perturbations to forbid any edge $ab$ passing through its endpoint $a$ in its interior, see Figure~\ref{fig:perturb_block}. Now we can remove touchings, unpack coincident segments, and resolve self-crossings as above. Since Condition (S) forbids non-adjacent edges to touch, the removal of touchings will not make two non-adjacent edges cross more than once. At the end, we obtain a separate-simple complete graph drawing satisfying the mild assumptions as wanted.

Finally, we verify the existence of non-sandwiched segments. Let $p$ be a touching point on a segment $\alpha$, and consider a sufficiently small closed disk $D$ centered at $p$ whose boundary circle meets the drawing transversely and contains no other touching or crossing points. The two endpoints of $\alpha\cap \partial D$ split $\partial D$ into two arcs; we fix one and call it the left arc of $\alpha$. For any segment $\beta$ that touches $\alpha$ at $p$, we say that $\beta$ touches $\alpha$ from the left if both endpoints of $\beta\cap \partial D$ lie on the left arc of $\alpha$. If no segment touches $\alpha$ from the left, then $\alpha$ is not sandwiched. Otherwise, we consider all segments touching $\alpha$ from the left. Such a segment $\beta$ determines its own left arc on $\partial D$ (namely, the arc of $\partial D$ cut out by $\beta\cap \partial D$ that lies entirely within the left arc of $\alpha$). We define a partial order by declaring $\gamma \succ \beta$ if the left arc of $\gamma$ is contained in the left arc of $\beta$. A maximal element of this partially ordered set cannot be sandwiched, hence such a non-sandwiched segment always exists.
\end{proof}

\section{Discussion}\label{sec:discussion}

It is natural to ask for a characterization of the plane structures guaranteed in simple complete graph drawings. To address this question, we can follow an approach similar to that of Theorem~\ref{thm:main2}. We call a simple complete graph drawing \textit{convex} if its vertices can be ordered as $v_1,v_2,\dots,v_n$ such that $v_i,v_j,v_k,v_\ell$ induce exactly the crossing $v_iv_k$ with $v_jv_\ell$ for all $i<j<k<\ell$, and we call it \textit{twisted} if its vertices can be ordered such that $v_i,v_j,v_k,v_\ell$ induce exactly the crossing $v_iv_\ell$ with $v_jv_k$ for all $i<j<k<\ell$. Pach, Solymosi, and T\'oth~\cite{pach2003unavoidable} have shown every sufficiently large simple complete graph drawing contains a sub-drawing on arbitrarily many vertices that is either convex or twisted. On the other hand, it is known both convex drawings and twisted drawings exist and can be arbitrarily large~\cite{harborth1992drawings}. Therefore, an abstract graph $G$ is a plane structure guaranteed in simple complete graph drawings if and only if $G$ is \textit{stack} and \textit{queue}. Here, $G$ being \textit{stack} (\textit{queue} respectively) means its vertices can be ordered such that $v_iv_k$ and $v_jv_\ell$ ($v_iv_\ell$ and $v_jv_k$ respectively) are not edges of $G$ simultaneously for all $i<j<k<\ell$. The names of these properties come from their relationships with the data structures stack and queue. We refer our readers to \cite{dujmovicW2004layouts,hlinveny2025stack,pupyrev2020layouts} and the references therein for information on related studies. It is known that stack graphs have nice characterizations~\cite{syslo79stack} and can be recognized in linear time, but the task of recognizing queue graphs is NP-complete~\cite{heathR92queues}. It seems to be an open problem whether graphs that are both stack and queue can be recognized efficiently, but that is beyond the scope of this paper.

\begin{figure}[ht]
    \centering
    \includegraphics{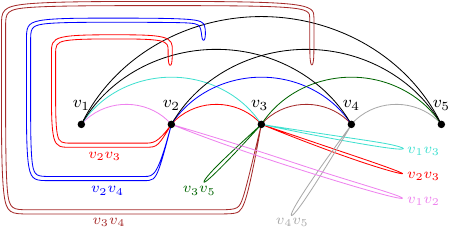}
    \caption{A loop labeled $v_iv_j$ is part of the edge between vertices $v_i$ and $v_j$.}
    \label{fig:selfblock}
\end{figure}

In Corollary~\ref{cor:main0}, we cannot further relax Condition (A) to allow adjacent edges to touch, nor relax Condition (S) to allow non-adjacent edges to touch. Otherwise, we pick vertices $p_1,p_2,\dots,p_n$ on a circle, and draw internally disjoint curves $h_{ij}$ ($j\neq i$) from $p_i$ to the center, then define edge $c_{ij}$ as the concatenation of $h_{ij}$ and $h_{ji}$. In such way, any two edges intersect exactly once in their interiors. If we choose the order for $h_{ij}$ to emanate from each vertex $p_i$ to be the same order as $p_i$ sees other vertices, then any two adjacent edges touch. If we choose the reverse order for each $p_i$, then any two adjacent edges cross; in fact, this is a degenerate case of the construction in the proof of Theorem~\ref{thm:main1} by choosing $r_1 = \dots = r_d = 0$. We cannot further relax the equation in Condition (A) to $|\text{Im}(c_{ij}) \cap \text{Im}(c_{k\ell})| = 1$, because we need to forbid edges from passing through their own endpoints. Otherwise, we can make thin loops emanating from vertices like open-ended half-edges and include them as edge parts, which gives us more flexibility to create crossings, see Figure~\ref{fig:selfblock} for a construction in which adjacent edges do not cross or touch but no two edges are disjoint. This construction can be extended to arbitrarily many vertices. In contrast, perhaps interestingly, edges can pass through endpoints in Condition (S) for Corollary~\ref{cor:main0}.

After all, Corollary~\ref{cor:main0} may not capture the exact cause of disjoint curves in the process of connecting points on paper. For example, we did not study the scenario that forbids non-adjacent curves to touch but allows adjacent curves to touch. In fact, we did not consider whether a rule like ``at least 64\% pairs of non-adjacent edges cross less than twice'' would force a pair of disjoint edges in complete graph drawings. We leave these obscure and interesting questions drifting in this river called time.

\nocite{aichholzer2024twisted_conference2022,bergold2024holes_conference,bergold2025plane_conference2024,fox2025structure_conference2024,pach2003unavoidable_conference2000,pach1994some_conference1993,pach2010disjoint_conference2003,sharir2013counting_conference2012,suk2013disjoint_conference2012,suk2025unavoidable_conference2022,felsner2020maximum_conference}
\bibliographystyle{plainurl}
\bibliography{main}

@article{hlinveny2025stack,
  title={Stack and queue numbers of graphs revisited},
  author={Hlin{\v{e}}n{\`y}, Petr and Straka, Adam},
  journal={European Journal of Combinatorics},
  volume={129},
  pages={104094},
  year={2025},
  publisher={Elsevier},
  doi={10.1016/j.ejc.2024.104094}
}

@article{harborth1992drawings,
  title={Drawings of the complete graph with maximum number of crossings},
  author={Harborth, Heiko and Mengersen, Ingrid},
  journal={Congressus Numerantium},
  pages={225--225},
  year={1992},
  publisher={UTILITAS MATHEMATICA PUBLISHING INC}
}

@article{hertrich2022coloring,
  title={Coloring Drawings of Graphs},
  author={Hertrich, Christoph and Schr{\"o}der, Felix and Steiner, Raphael},
  journal={The Electronic Journal of Combinatorics},
  volume={29},
  number={1},
  pages={\#P1.17},
  year={2022},
  doi = {10.37236/9808}
}

@book{diestel2025graph,
  title={Graph theory},
  author={Diestel, Reinhard},
  year={2025},
  publisher={Springer Nature}
}

@article{ramsey1930problem,
  title={On a Problem of Formal Logic},
  author={Ramsey, Frank P.},
  journal={Proceedings of the London Mathematical Society},
  volume={2},
  number={1},
  pages={264--286},
  year={1930},
  doi={10.1112/plms/s2-30.1.264},
}

@phdthesis{rafla1988,
  title={The good drawings of the complete graph},
  author={Rafla, Nabil H},
  year={1988},
  school={McGill University},
}

@article{pach2003unavoidable,
  title={Unavoidable configurations in complete topological graphs},
  author={Pach, J{\'a}nos and Solymosi, J{\'o}zsef and T{\'o}th, G{\'e}za},
  journal={Discrete \& Computational Geometry},
  volume={30},
  number={2},
  pages={311--320},
  year={2003},
  _publisher={Springer},
  doi={10.1007/S00454-003-0012-9},
  note={\textit{Conference Version:~\cite{pach2003unavoidable_conference2000}.}},
}

@inproceedings{pach2003unavoidable_conference2000,
  title={Unavoidable configurations in complete topological graphs},
  author={Pach, J{\'a}nos and T{\'o}th, G{\'e}za},
  booktitle={Proceedings of the 8th International Symposium on Graph Drawing ({GD} 2000)},
  pages={328--337},
  _publisher={Springer},
  year={2000},
  doi={10.1007/3-540-44541-2_31},
}

@article{aichholzer2024twisted,
  title={Twisted ways to find plane structures in simple drawings of complete graphs},
  author={Aichholzer, Oswin and Garc{\'\i}a, Alfredo and Tejel, Javier and Vogtenhuber, Birgit and Weinberger, Alexandra},
  journal={Discrete \& Computational Geometry},
  volume={71},
  number={1},
  pages={40--66},
  year={2024},
  _publisher={Springer},
  doi={10.1007/S00454-023-00610-0},
  note={\textit{Conference Version:~\cite{aichholzer2024twisted_conference2022}.}},
}

@inproceedings{aichholzer2024twisted_conference2022,
  author       = {Oswin Aichholzer and
                  Alfredo Garc{\'{\i}}a and
                  Javier Tejel and
                  Birgit Vogtenhuber and
                  Alexandra Weinberger},
  title        = {Twisted Ways to Find Plane Structures in Simple Drawings of Complete
                  Graphs},
  booktitle    = {Proceedings of the 38th International Symposium on Computational Geometry ({SoCG} 2022)},
  pages        = {5:1--5:18},
  _publisher    = {Schloss Dagstuhl - Leibniz-Zentrum f{\"{u}}r Informatik},
  year         = {2022},
  doi          = {10.4230/LIPICS.SOCG.2022.5},
}

@article{suk2025unavoidable,
  title={Unavoidable patterns in complete simple topological graphs},
  author={Suk, Andrew and Zeng, Ji},
  journal={Discrete \& Computational Geometry},
  volume={73},
  number={1},
  pages={79--91},
  year={2025},
  _publisher={Springer},
  doi={10.1007/S00454-024-00658-6},
  note={\textit{Conference Version:~\cite{suk2025unavoidable_conference2022}.}},
}

@inproceedings{suk2025unavoidable_conference2022,
  author       = {Andrew Suk and
                  Ji Zeng},
  title        = {Unavoidable Patterns in Complete Simple Topological Graphs},
  booktitle    = {Proceedings of the 30th International Symposium on Graph Drawing and Network Visualization ({GD} 2022)},
  pages        = {3--15},
  _publisher    = {Springer},
  year         = {2022},
  doi          = {10.1007/978-3-031-22203-0_1},
}

@article{fox2025structure,
  title={A structure theorem for pseudosegments and its applications},
  author={Fox, Jacob and Pach, J{\'a}nos and Suk, Andrew},
  journal={Journal of Combinatorial Theory, Series B},
  volume={174},
  pages={99--132},
  year={2025},
  _publisher={Elsevier},
  doi={10.1016/J.JCTB.2025.04.007},
  note={\textit{Conference Version:~\cite{fox2025structure_conference2024}.}}, 
}

@inproceedings{fox2025structure_conference2024,
  title={A Structure Theorem for Pseudo-Segments and Its Applications},
  author={Fox, Jacob and Pach, J{\'a}nos and Suk, Andrew},
  booktitle={Proceedings of the 40th International Symposium on Computational Geometry (SoCG 2024)},
  pages={59:1--59:14},
  year={2024},
  organization={Schloss Dagstuhl--Leibniz-Zentrum f{\"u}r Informatik},
  doi={10.4230/LIPICS.SOCG.2024.59},
}

@article{suk2013disjoint,
  author       = {Andrew Suk},
  title        = {Disjoint Edges in Complete Topological Graphs},
  journal      = {Discrete \& Computational Geometry},
  volume       = {49},
  number       = {2},
  pages        = {280--286},
  year         = {2013},
  doi          = {10.1007/S00454-012-9481-X},
  note={\textit{Conference Version:~\cite{suk2013disjoint_conference2012}.}},
}

@inproceedings{suk2013disjoint_conference2012,
  title={Disjoint edges in complete topological graphs},
  author={Suk, Andrew},
  booktitle={Proceedings of the twenty-eighth annual symposium on Computational geometry},
  pages={383--386},
  year={2012},
  doi={10.1145/2261250.2261308},
}

@article{garcia2021plane,
  title={{On plane subgraphs of complete topological drawings}},
  author={Garc{\'i}a, Alfredo and Pilz, Alexander and Tejel, Javier},
  journal={Ars Mathematica Contemporanea},
  volume={20},
  number={1},
  pages={69--87},
  year={2021},
  doi={10.26493/1855-3974.2226.E93},
}

@inproceedings{fulek2013topological,
  title={Topological graphs: empty triangles and disjoint matchings},
  author={Fulek, Radoslav and Ruiz-Vargas, Andres J},
  booktitle={Proceedings of the twenty-ninth annual symposium on Computational geometry},
  pages={259--266},
  year={2013},
  doi={10.1145/2462356.2462394},
}

@article{ruizvargas2017disjoint,
  author       = {Andres J. Ruiz{-}Vargas},
  title        = {Many disjoint edges in topological graphs},
  journal      = {Computational Geometry},
  volume       = {62},
  pages        = {1--13},
  year         = {2017},
  doi          = {10.1016/J.COMGEO.2016.11.003},
}

@article{fulek14disjoint,
  author       = {Radoslav Fulek},
  title        = {Estimating the Number of Disjoint Edges in Simple Topological Graphs
                  via Cylindrical Drawings},
  journal      = {{SIAM} Journal on Discrete Mathematics},
  volume       = {28},
  number       = {1},
  pages        = {116--121},
  year         = {2014},
  doi          = {10.1137/130925554},
}

@article{bergold2024holes,
  author       = {Helena Bergold and
                  Joachim Orthaber and
                  Manfred Scheucher and
                  Felix Schr{\"{o}}der},
  title        = {Holes in Convex and Simple Drawings},
  journal      = {Journal of Graph Algorithms and Applications},
  volume       = {29},
  number       = {3},
  pages        = {23--38},
  year         = {2025},
  doi          = {10.7155/JGAA.V29I3.2999},
  note={\textit{Conference Version:~\cite{bergold2024holes_conference}.}},
}

@inproceedings{bergold2024holes_conference,
  title={Holes in Convex and Simple Drawings},
  author={Bergold, Helena and Orthaber, Joachim and Scheucher, Manfred and Schr{\"o}der, Felix},
  booktitle={Proceedings of the 32nd International Symposium on Graph Drawing and Network Visualization (GD 2024)},
  pages={5:1--5:9},
  year={2024},
  organization={Schloss Dagstuhl--Leibniz-Zentrum f{\"u}r Informatik},
  doi={10.4230/LIPICS.GD.2024.5},
}

@inproceedings{aichholzer2022shooting,
  title={{Shooting stars in simple drawings of $K_{m,n}$}},
  author={Aichholzer, Oswin and Garc{\'\i}a, Alfredo and Parada, Irene and Vogtenhuber, Birgit and Weinberger, Alexandra},
  booktitle={Proceedings of the 30th International Symposium on Graph Drawing and Network Visualization ({GD} 2022)},
  pages={49--57},
  year={2022},
  organization={Springer},
  doi={10.1007/978-3-031-22203-0_5},
}

@article{zeng2025note,
  title={Note on disjoint faces in simple topological graphs},
  author={Zeng, Ji},
  journal={Journal of Combinatorial Theory, Series B},
  volume={171},
  pages={28--35},
  year={2025},
  _publisher={Elsevier},
  doi={10.1016/J.JCTB.2024.11.002},
}

@article{pach2010disjoint,
  title={Disjoint edges in topological graphs},
  author={Pach, J{\'a}nos and T{\'o}th, G{\'e}za},
  journal={Journal of Combinatorics},
  volume={1},
  pages={335--344},
  year={2010},
  doi={10.4310/JOC.2010.v1.n3.a4},
  note={\textit{Conference Version:~\cite{pach2010disjoint_conference2003}.}},
}

@inproceedings{pach2010disjoint_conference2003,
  author       = {J{\'{a}}nos Pach and
                  G{\'{e}}za T{\'{o}}th},
  title        = {Disjoint Edges in Topological Graphs},
  booktitle    = {Proceedings of the Indonesia-Japan Joint Conference on Combinatorial Geometry and Graph Theory ({IJCCGGT
                  2003})},
  pages        = {133--140},
  _publisher    = {Springer},
  year         = {2003},
  doi          = {10.1007/978-3-540-30540-8_15},
}

@article{pach1994some,
  author       = {J{\'{a}}nos Pach and
                  Jen{\H{o}} T{\H{o}}r{\H{o}}csik},
  title        = {Some Geometric Applications of {D}ilworth's Theorem},
  journal      = {Discrete \& Computational Geometry},
  volume       = {12},
  pages        = {1--7},
  year         = {1994},
  doi          = {10.1007/BF02574361},
  note={\textit{Conference Version:~\cite{pach1994some_conference1993}.}},
}

@inproceedings{pach1994some_conference1993,
  title={{Some geometric applications of Dilworth's theorem}},
  author={Pach, J{\'a}nos and T{\H{o}}r{\H{o}}csik, Jen{\H{o}}},
  booktitle={Proceedings of the ninth annual symposium on Computational geometry},
  pages={264--269},
  year={1993},
  doi={10.1007/BF02574361},
}

@article{toth2000note,
  title={Note on geometric graphs},
  author={T{\'o}th, G{\'e}za},
  journal={Journal of Combinatorial Theory, Series A},
  volume={89},
  number={1},
  pages={126--132},
  year={2000},
  doi={10.1006/JCTA.1999.3001},
}

@article{sharir2013counting,
  author       = {Micha Sharir and
                  Adam Sheffer and
                  Emo Welzl},
  title        = {Counting plane graphs: Perfect matchings, spanning cycles, and {K}asteleyn's
                  technique},
  journal      = {Journal of Combinatorial Theory, Series A},
  volume       = {120},
  number       = {4},
  pages        = {777--794},
  year         = {2013},
  doi          = {10.1016/J.JCTA.2013.01.002},
  note={\textit{Conference Version:~\cite{sharir2013counting_conference2012}.}},
}

@inproceedings{sharir2013counting_conference2012,
  title={{Counting plane graphs: perfect matchings, spanning cycles, and Kasteleyn's technique}},
  author={Sharir, Micha and Sheffer, Adam and Welzl, Emo},
  booktitle={Proceedings of the twenty-eighth annual symposium on Computational geometry},
  pages={189--198},
  year={2012},
  doi={10.1145/2261250.2261277},
}

@inproceedings{aichholzer2024separable,
  title={Separable Drawings: Extendability and Crossing-Free {H}amiltonian Cycles},
  author={Aichholzer, Oswin and Orthaber, Joachim and Vogtenhuber, Birgit},
  booktitle={Prooceedings of the 32nd International Symposium on Graph Drawing and Network Visualization (GD 2024)},
  pages={34:1--34:17},
  year={2024},
  organization={Schloss Dagstuhl--Leibniz-Zentrum f{\"u}r Informatik},
  doi={10.4230/LIPICS.GD.2024.34},
}

@article{bergold2025plane,
  title={Plane {H}amiltonian cycles in convex drawings},
  author       = {Helena Bergold and
                  Stefan Felsner and
                  Meghana M. Reddy and
                  Joachim Orthaber and
                  Manfred Scheucher},
  journal={Discrete \& Computational Geometry},
  year={2025},
  _publisher={Springer},
  doi={https://doi.org/10.1007/s00454-025-00752-3},
  note={\textit{Conference Version:~\cite{bergold2025plane_conference2024}.}},
}

@inproceedings{bergold2025plane_conference2024,
  author       = {Helena Bergold and
                  Stefan Felsner and
                  Meghana M. Reddy and
                  Joachim Orthaber and
                  Manfred Scheucher},
  title        = {Plane {H}amiltonian Cycles in Convex Drawings},
  booktitle    = {Proceedings of the 40th International Symposium on Computational Geometry (SoCG 2024)},
  pages        = {18:1--18:16},
  _publisher    = {Schloss Dagstuhl - Leibniz-Zentrum f{\"{u}}r Informatik},
  year         = {2024},
  doi          = {10.4230/LIPICS.SOCG.2024.18},
}

@inproceedings{felsner2020maximum_conference,
  author       = {Stefan Felsner and
                  Michael Hoffmann and
                  Kristin Knorr and
                  Irene Parada},
  title        = {On the Maximum Number of Crossings in Star-Simple Drawings of {$K_n$}
                  with No Empty Lens},
  booktitle    = {Proceedings of 28th International Symposium on Graph Drawing and Network Visualization (GD 2020)},
  pages        = {382--389},
  _publisher    = {Springer},
  year         = {2020},
  doi          = {10.1007/978-3-030-68766-3_30},
}

@article{felsner2020maximum,
  author       = {Stefan Felsner and
                  Michael Hoffmann and
                  Kristin Knorr and
                  Jan Kyn\v{c}l and
                  Irene Parada},
  title        = {On the Maximum Number of Crossings in Star-Simple Drawings of {$K_n$}
                  with No Empty Lens},
  journal      = {Journal of Graph Algorithms and Applications},
  volume       = {26},
  number       = {3},
  pages        = {381--399},
  year         = {2022},
  doi          = {10.7155/jgaa.00600},
  note={\textit{Conference Version:~\cite{felsner2020maximum_conference}.}},
}

@article{cairns2019bad,
  author       = {Grant Cairns and
                  Emily Groves and
                  Yuri Nikolayevsky},
  title        = {Bad drawings of small complete graphs},
  journal      = {Australasian Journal of Combinatorics},
  volume       = {75},
  pages        = {322--342},
  year         = {2019}
}

@article{balko2015semi,
  author       = {Martin Balko and
                  Radoslav Fulek and
                  Jan Kyn\v{c}l},
  title        = {Crossing Numbers and Combinatorial Characterization of Monotone Drawings
                  of {$K_n$}},
  journal      = {Discrete \& Computational Geometry},
  volume       = {53},
  number       = {1},
  pages        = {107--143},
  year         = {2015},
  doi          = {10.1007/S00454-014-9644-Z},
}

@article{heathR92queues,
  author       = {Lenwood S. Heath and
                  Arnold L. Rosenberg},
  title        = {Laying out Graphs Using Queues},
  journal      = {{SIAM} Journal on Computing},
  volume       = {21},
  number       = {5},
  pages        = {927--958},
  year         = {1992},
  doi          = {10.1137/0221055},
}

@article{syslo79stack,
  author       = {Maciej M. Syslo},
  title        = {Characterizations of outerplanar graphs},
  journal      = {Discrete Mathematics},
  volume       = {26},
  number       = {1},
  pages        = {47--53},
  year         = {1979},
  doi          = {10.1016/0012-365X(79)90060-8},
}

@article{dujmovicW2004layouts,
  author       = {Vida Dujmovic and
                  David R. Wood},
  title        = {On Linear Layouts of Graphs},
  journal      = {Discret. Math. Theor. Comput. Sci.},
  volume       = {6},
  number       = {2},
  pages        = {339--358},
  year         = {2004},
  doi          = {10.46298/DMTCS.317},
}

@misc{pupyrev2020layouts,
  author       = {Pupyrev, Sergey},
  title        = {A collection of existing results on stack and queue numbers},
  howpublished = {\url{https://spupyrev.github.io/linearlayouts.html}},
  year         = {2020},
  note         = {Accessed: 2/19/2026}
}

\end{document}